\documentclass[10pt, draftclsnofoot, onecolumn]{IEEEtran}
\usepackage{amsthm}
\usepackage{amsmath}
\usepackage{amsfonts}
\usepackage{graphicx}
\usepackage{latexsym}
\usepackage{amssymb}
\usepackage{stmaryrd}
\usepackage{comment}
\usepackage{amscd}
\usepackage{hyperref}
\newcommand\myshade{70} 
\hypersetup{ 
	linkcolor  = red!\myshade!black,
	citecolor  = blue!\myshade!black,
	urlcolor   = blue!\myshade!black,
	colorlinks = true,
}
\usepackage{cite}
\usepackage{cleveref}
\usepackage[belowskip=-15pt,aboveskip=0pt,small]{caption}
\usepackage[dvipsnames]{xcolor}
\usepackage[linesnumbered,ruled]{algorithm2e}
\usepackage[noend]{algpseudocode}
\allowdisplaybreaks
\usepackage{graphicx}
\usepackage{subfigure}
\usepackage{wrapfig}
\usepackage{float}
\usepackage{bm}
\usepackage{bbm}
\usepackage{enumerate}
\usepackage{thm-restate}

\usepackage{enumitem}

\usepackage{mathtools}
\newcommand{\bea}{\begin{eqnarray}}
\newcommand{\eea}{\end{eqnarray}}
\newcommand{\bean}{\begin{eqnarray*}}
\newcommand{\eean}{\end{eqnarray*}}

\newcommand{\sbinom}[2]{\left( \begin{array}{c} #1 \\ #2 \end{array} \right) }

\newcommand{\Perr}{\mathsf{P_{err}}}

\newcommand{\field}[1]{\mathbb{#1}}

\newcommand{\F}{\field{F}}

\newcommand{\cC}{{\cal C}}

\newcommand{\cL}{{\cal L}}

\newcommand{\cO}{{\cal O}}

\newcommand{\cS}{{\cal S}}

\newcommand{\cU}{{\cal U}}

\newcommand{\sG}{\script{G}}

\newcommand{\sP}{\script{P}}

\newcommand{\bfc}{{\boldsymbol c}}

\newcommand{\bfe}{{\boldsymbol e}}

\newcommand{\bfv}{{\boldsymbol v}}

\newcommand{\bfx}{{\boldsymbol x}}
\newcommand{\bfy}{{\boldsymbol y}}

\newcommand{\bfX}{{\mathbf X}}
\newcommand{\bfY}{{\mathbf Y}}

\newcommand{\FF}{\mathbb{F}}

\DeclareMathOperator*{\argmax}{arg\,max}

\DeclareMathAlphabet{\mathbfsl}{OT1}{cmr}{bx}{it}
\newcommand{\uuu}{\kern-1pt\mathbfsl{u}\kern-0.5pt}
\newcommand{\vvv}{\kern-1pt\mathbfsl{v}\kern-0.5pt}

\newcommand{\myboxplus}{\kern1pt\mbox{\small$\boxplus$}}

\makeatletter \DeclareRobustCommand{\sbinom}{\genfrac[]\z@{}}
\makeatother
\newcommand{\G}[2]{\sbinom{{#1}\kern-1pt}{{#2}\kern-1pt}}
\newcommand{\Gq}[2]{\sbinom{{#1}\kern-0.25pt}{{#2}\kern-0.25pt}}
\newcommand{\Fq}{\smash{{\mathbb F}_{\!q}}}

\newcommand{\Ps}{\smash{{\sP\kern-2.0pt}_q\kern-0.5pt(n)}}
\newcommand{\sPs}{\smash{{\sP\kern-1.5pt}_q(n)}}
\newcommand{\Ptwo}{\smash{{\sP\kern-2.0pt}_2\kern-0.5pt(n)}}
\newcommand{\Ptwom}{\smash{{\sP\kern-2.0pt}_2\kern-0.5pt(m)}}
\newcommand{\Ptwonm}{\smash{{\sP\kern-2.0pt}_2\kern-0.5pt(n+m)}}
\newcommand{\Ptwoa}{\smash{{\sP\kern-2.0pt}_2\kern-0.5pt(1)}}
\newcommand{\Ptwob}{\smash{{\sP\kern-2.0pt}_2\kern-0.5pt(2)}}
\newcommand{\Ptwoc}{\smash{{\sP\kern-2.0pt}_2\kern-0.5pt(3)}}
\newcommand{\Ptwod}{\smash{{\sP\kern-2.0pt}_2\kern-0.5pt(4)}}
\newcommand{\Ptwoe}{\smash{{\sP\kern-2.0pt}_2\kern-0.5pt(5)}}
\newcommand{\Ptwof}{\smash{{\sP\kern-2.0pt}_2\kern-0.5pt(6)}}
\newcommand{\Ptwokm}{\smash{{\sP\kern-2.0pt}_2\kern-0.5pt(2k-1)}}
\newcommand{\Pone}{\smash{{\sP\kern-2.5pt}_2\kern-0.5pt(n{-}1)}}

\newcommand{\Gr}{\smash{{\sG\kern-1.5pt}_q\kern-0.5pt(n,k)}}
\newcommand{\Gi}{\smash{{\sG\kern-1.5pt}_q\kern-0.5pt(n,i)}}
\newcommand{\Gj}{\smash{{\sG\kern-1.5pt}_q\kern-0.5pt(n,j)}}
\newcommand{\Grmk}{\smash{{\sG\kern-1.5pt}_q\kern-0.5pt(n,n-k)}}
\newcommand{\Grdk}{\smash{{\sG\kern-1.5pt}_q\kern-0.5pt(2k,k)}}
\newcommand{\Grekappa}{\smash{{\sG\kern-1.5pt}_q\kern-0.5pt(n,e+1-\kappa)}}
\newcommand{\Grtwoekappa}{\smash{{\sG\kern-1.5pt}_q\kern-0.5pt(n,2e+1-\kappa)}}
\newcommand{\Gremkappa}{\smash{{\sG\kern-1.5pt}_q\kern-0.5pt(n,e-\kappa)}}
\newcommand{\Gn}{\smash{{\sG\kern-1.5pt}_2\kern-0.5pt(n,n{-}1)}}
\newcommand{\Gnq}{\smash{{\sG\kern-1.5pt}_q\kern-0.5pt(n,n{-}1)}}
\newcommand{\Gone}{\smash{{\sG\kern-1.5pt}_2\kern-0.5pt(n,1)}}
\newcommand{\Gqone}{\smash{{\sG\kern-1.5pt}_q\kern-0.5pt(n,1)}}
\newcommand{\GTwo}{\smash{{\sG\kern-1.5pt}_2\kern-0.5pt(n,k)}}
\newcommand{\GTwonk}[2]{{\smash{{\sG\kern-1.5pt}_2\kern-0.5pt({#1},{#2})}}}
\newcommand{\Gnk}{\smash{{\sG\kern-1.5pt}_2\kern-0.5pt(n,n{-}k)}}
\newcommand{\Greone}{\smash{{\sG\kern-1.5pt}_q\kern-0.5pt(n,e{+}1)}}
\newcommand{\Gretwo}{\smash{{\sG\kern-1.5pt}_q\kern-0.5pt(n,e{+}2)}}

\newcommand{\be}[1]{\begin{equation}\label{#1}}
\newcommand{\ee}{\end{equation}}

\newcommand{\Eras}{\mathsf{Erasure}}
\newcommand{\Err}{\mathsf{Error}}
\newcommand{\Suc}{\mathsf{Success}}

\newcommand{\NEras}{\mathsf{Num_{Erasure}}}
\newcommand{\NErr}{\mathsf{Num_{Error}}}
\newcommand{\NSuc}{\mathsf{Num_{Success}}}

\newcommand{\tf}{\tfrac}

\newtheorem{theorem}{Theorem}
\newtheorem{lemma}{Lemma}

\newtheorem{corollary}{Corollary}

\newtheorem{definition}{Definition}

\newtheorem*{InfTheorem}{Informal Theorem}

\begin{document}

\author{\IEEEauthorblockN{ {Shubhransh~Singhvi}\IEEEauthorrefmark{3}, {Han~Mao~Kiah}\IEEEauthorrefmark{2} and {Eitan~Yaakobi}\IEEEauthorrefmark{3}}\\
\IEEEauthorblockA{\IEEEauthorrefmark{3}%
                     Department of Computer Science, 
                     Technion, Israel}\\  \IEEEauthorblockA{\IEEEauthorrefmark{2}%
                     School of Physical and Mathematical Sciences, 
		NTU, Singapore}\qquad
    \IEEEauthorblockA{
    \\Email: shubhranshsinghvi2001@gmail.com, HMKiah@ntu.edu.sg, yaakobi@cs.technion.ac.il}
 }

\title{Reconstructing Reed--Solomon Codes \\from Multiple Noisy Channel Outputs}
\date{\today}
 \maketitle
\thispagestyle{empty}	
\begin{abstract}
The \emph{sequence reconstruction problem}, introduced by Levenshtein in 2001, considers a communication setting in which a sender transmits a codeword and the receiver observes $K$ independent noisy versions of this codeword. In this work, we study the problem of \emph{efficient} reconstruction when each of the $K$ outputs is corrupted by a $q$-ary discrete memoryless symmetric (DMS) substitution channel with substitution probability $p$. Focusing on Reed--Solomon (RS) codes, we adapt the Koetter--Vardy soft-decision decoding algorithm to obtain an efficient reconstruction algorithm. For sufficiently large blocklength and alphabet size, we derive an explicit rate threshold, depending only on $(p,K)$, such that the transmitted codeword can be reconstructed with arbitrarily small probability of error whenever the code rate $R$ lies below this threshold.

\end{abstract}





\section{Introduction}
The \emph{sequence reconstruction problem}, introduced by Levenshtein~\cite{L01A,L01B}, considers a communication model in which a codeword is transmitted through multiple noisy channels, and the receiver observes several corrupted versions of the same input. A central question in this framework is to determine the minimum number of channel outputs required to uniquely reconstruct the transmitted codeword. Levenshtein showed that, in the worst case, for unique reconstruction, the number of channel outputs must be greater than the maximum intersection size between the sets of possible outputs (balls) corresponding to any two distinct channel inputs. Beyond uniqueness, an equally important challenge is the design of \emph{efficient decoding algorithms} that can successfully reconstruct the codeword from the noisy observations. While the problem itself was introduced in~\cite{L01A}, the development of efficient decoders has received comparatively less attention; see, for instance,~\cite{YB18,AY21,PGK22,JLL23,SCKY24}.

Originally motivated by applications in biology and chemistry, the sequence reconstruction model is also relevant to wireless sensor networks. More recently, it has attracted significant interest due to its applicability to DNA storage systems, where the same encoded DNA strand is sequenced multiple times, resulting in multiple noisy copies, or \emph{reads}, of the original strand~\cite{CGK12,Getal13,YKGMZM15}. Practical DNA storage systems often use concatenated code constructions, where the inner code introduces redundancy within each strand to correct symbol-level insertion/deletion/substitution errors, while the outer code adds redundancy across strands to recover from strand losses and correct residual errors left uncorrected by the inner code.
Maximum distance separable (MDS) codes, most notably Reed–Solomon (RS) codes \cite{RS60}, are widely used in practice as outer codes due to their optimal erasure and error correction capabilities \cite{Grass2015ChemicalPreservation, Erlich2017DNAFountain, Organick2018RandomAccess, Press2020HEDGES}. This motivates the study of efficient reconstruction of RS codewords from multiple noisy reads corrupted by substitution errors. We next give a definition of RS codes: 
 \begin{definition}\label{RS}
		Let $\alpha_1, \alpha_2, \ldots, \alpha_n \in\F_q$ be distinct points in a finite field $\mathbb{F}_q$ of order  $q\geq n$. For $k\leq n$, the $[n,k]_q$ RS code,  
		defined by the evaluation vector $\bm \alpha = ( \alpha_1, \ldots, \alpha_n )$, is the set of codewords 
        \[\left \lbrace c_f = \left( f(\alpha_1), \ldots, f(\alpha_n) \right) \mid f\in \mathbb{F}_q[x],\deg f < k  \right \rbrace\;.\] 
\end{definition}

In this work, we study the problem of reconstructing a transmitted RS codeword from $K$ independent noisy reads. Each read is corrupted by a $q$-ary discrete memoryless symmetric (DMS) substitution channel with substitution probability $p$. Given the multiset of $K$ received channel outputs, the objective is to reconstruct the transmitted RS codeword in polynomial time with vanishing probability of error as the blocklength grows.

\subsection{Prior work}
Most of the works on sequence reconstruction focus on worst-case uniqueness guarantees, while the design and analysis of efficient reconstruction algorithms under probabilistic noise models has received comparatively little attention. In~\cite{SCKY24}, the authors introduced a soft-decision reconstruction framework for RS codes under a \emph{bounded adversarial substitution model}. The approach aggregates soft reliability information across multiple reads into a carefully designed multiplicity matrix and applies the Koetter--Vardy soft-decision decoding algorithm. This reconstruction algorithm achieves input-linear time complexity and is capable of correcting substitution errors beyond the Johnson radius. Although the bounded-error model yields strong worst-case guarantees, it does not characterize typical-case behavior under natural probabilistic noise models.

\vspace{-0.2cm}
\subsection{Our Contributions}
We study the problem of reconstructing RS codewords under a
\emph{probabilistic $K$-draw $q$-ary DMS substitution channel model}.
The receiver observes $K$ independent noisy reads of the same RS codeword, and the
objective is efficient reconstruction with high probability.

In contrast to existing probabilistic reconstruction approaches that rely on hard
decisions or majority decoding, our framework explicitly exploits soft
reliability information extracted across multiple reads. We show that symbol positions with ambiguous observations, often treated as erasures
in hard-decision schemes, provide useful decoding information when incorporated through
a soft-decision multiplicity assignment. As in~\cite{SCKY24}, we aggregate the $K$ reads into a carefully designed multiplicity
matrix and apply the Koetter--Vardy soft-decision decoding algorithm.
Unlike the bounded adversarial setting of~\cite{SCKY24}, this work analyzes
the multiplicity construction and decoding conditions under a probabilistic noise
model, requiring new tools to characterize typical decoding behavior. We provide an explicit finite-length probabilistic analysis guaranteeing decoding
success with high probability. By asymptotically simplifying these conditions, we derive a clean achievable-rate threshold for reliable sequence reconstruction as the blocklength grows.

\begin{InfTheorem}[Informal statement of Theorems~3 and~4]
Consider a RS code of rate $R$ transmitted over a $K$-draw $q$-ary DMS
channel with substitution probability $p$.
There exists a polynomial-time reconstruction algorithm based on Koetter--Vardy
soft-decision decoding such that, for sufficiently large blocklength and alphabet
size, the transmitted codeword can be reconstructed with arbitrarily small error
probability whenever $R$ lies below an explicit threshold depending only on $(p,K)$.
\end{InfTheorem}

Fig.~\ref{fig:rateRegion-comp} illustrates the achievable-rate region and shows significant gains over hard-decision majority decoding. It also shows the capacity value in each case.

\begin{figure}
    \centering
    \includegraphics[scale=0.4
    ]{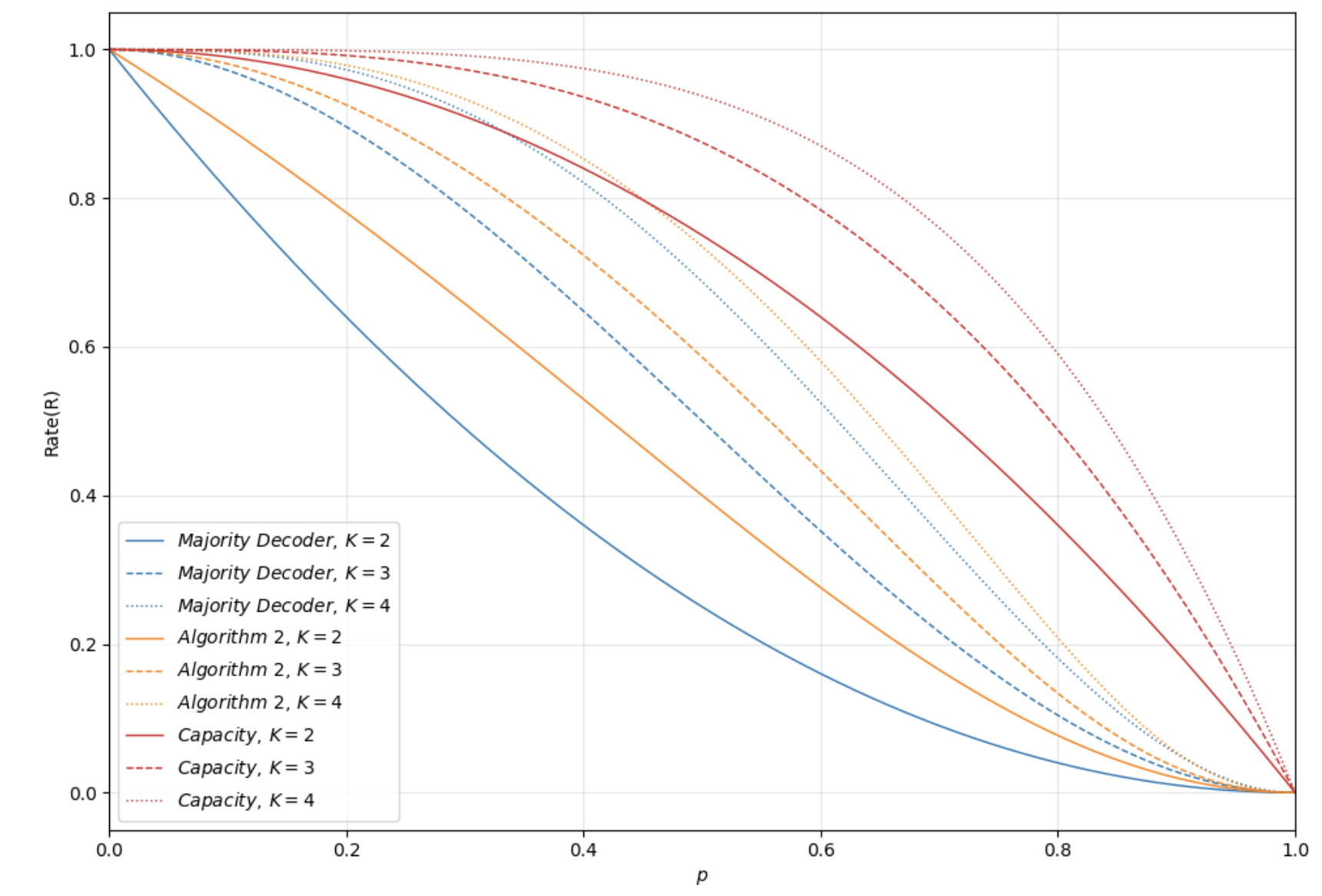}
    \caption{Comparison of achievable rate regions of Algorithm~\ref{alg:read-rec-K-reads} (Theorem~\ref{thm:R_upBound_DMS}), hard-decision majority-decoder (Corollary~\ref{cor:majority}) and channel capacity (Theorem~\ref{thm:DMSqKpCapacity}) for $K=2$, $K=3$ and $K=4$. 
    }
    \label{fig:rateRegion-comp}
\end{figure}

\subsection{Related Work}
Solving the reconstruction problem was studied in~\cite{L01A} with respect to several channels such as the Hamming distance, the Johnson graphs, and other metric distances. In~\cite{K08,K07,KLS07}, it was analyzed for permutations, and in~\cite{LKKM08,LS09} for other general error graphs. The problem was further investigated in~\cite{YSLB13} for permutations under the Kendall’s $\tau$ distance and the Grassmann graph. For insertion and deletion channels, a substantial body of work has developed various reconstruction guarantees and bounds; see, for example,~\cite{GY16,SGSD17,PGK22,sun23,lan25,pham25}. Connections between the reconstruction problem and associative memory models have also been established in several works, including~\cite{JL14,JL15,JL16,YB18}. The problem was also studied in~\cite{JV04} in the context of asymptotically improving the Gilbert–Varshamov bound. Complementary to these lines of work,~\cite{WangYaakobiZhang2024SimpleConditions} analyzes the expected number of channel outputs required for unique reconstruction under bounded adversarial substitution channels, while~\cite{PapadopoulouRameshwarWachterZeh2024Views} derives sufficient conditions for unique reconstruction in the same setting.

\subsection{Soft Decoding \'a la Koetter and Vardy}

Koetter and Vardy~\cite{KV03} extended the Guruswami--Sudan (GS) list-decoding algorithm~\cite{GS99} by incorporating probabilistic reliability information associated with the received symbols. This extension enables \emph{soft-decision decoding} by allowing unequal multiplicities to be assigned to interpolation points based on their reliabilities. A convenient way to represent the interpolation points and their assigned multiplicities is through a \emph{multiplicity matrix}.

\begin{definition}[\hspace{-0.1ex}\cite{KV03}] 
Let $\delta_0,\delta_1, \ldots, \delta_{q-1}$ be some ordering of $\F_q$.
A multiplicity matrix, denoted by $M$, is a $(q \times n)$-matrix with entries $m_{i,j}$ denoting the multiplicity of $(\delta_i, \alpha_j)$.
\end{definition}

We provide a high-level description of the Koetter--Vardy (KV) soft-decision decoding algorithm, following the exposition in~\cite{SCKY24}. Given a multiplicity matrix $M$, the KV algorithm computes a nonzero bivariate polynomial $Q_M(X,Y)$ of minimum $(1,k-1)$-weighted degree such that $Q_M(X,Y)$ has a zero of multiplicity at least $m_{i,j}$ at each point $(\alpha_j, \delta_i)$ for which $m_{i,j} \neq 0$. The polynomial $Q_M(X,Y)$ is then factorized to produce a list of candidate codewords \cite{KV03}. The \emph{cost} of constructing $Q_M(X,Y)$ from a multiplicity matrix $M$, denoted by $C(M)$, corresponds to the number of linear constraints imposed during the interpolation step and is defined as follows.

\vspace{-0.2cm}
\begin{definition}[\hspace{-0.1ex}\cite{KV03}] 
\label{def:cost}
The cost for a $q\times n$ multiplicity matrix $M$ is defined as
$C(M) := \sum_{i=0}^{q-1}\sum_{j=0}^{n-1}\binom{m_{i,j}+1}{2}$.
\end{definition}

\vspace{-0.2cm}
For \( \bfv \in \FF_q^{n} \), let \( [\bfv] \) denote the \( (q \times n) \)-matrix representation of \( \bfv \), where \( [\bfv]_{i,j} = 1 \) if \( v_j = \delta_i \), and \( [\bfv]_{i,j} = 0 \) otherwise. 
\vspace{-0.2cm}

\begin{definition}[\hspace{-0.1ex}\cite{KV03}] 
    The score of a vector \( \bfv \in \FF_q^{n} \) with respect to a multiplicity matrix \( M \) is defined as the inner product $\cS_M(\bfv) := \left<M, [\bfv] \right>$.
\end{definition}
\vspace{-0.2cm}
With these definitions in place, we state the main performance guarantee of the Koetter--Vardy (KV) algorithm.
\vspace{-0.2cm}

\begin{theorem}[\hspace{-0.1ex}\cite{KV03}]
\label{thm:kv-alg}
Let $\cC$ be an $[n,k]_q$ RS code defined over evaluation points $\bm{\alpha} = (\alpha_1,\ldots,\alpha_n)$, and let $M$ be a $q \times n$ multiplicity matrix. Given $M$ as input, the KV algorithm outputs a list $\cL$ with the following properties:
\begin{enumerate}
    \item For any codeword \( \bfc \in \cC \), if \( \cS_M(\bfc) \geq \sqrt{2(k-1)C(M)} \), then $\bfc$ is included in the list \( \cL \). 
    \item The list size is bounded as $|\cL| \le \sqrt{\frac{2C(M)}{k-1}}$.
    \item The algorithm runs in time $\cO\!\left((C(M))^3\right)$.
\end{enumerate}
\end{theorem}

\section{Sequence Reconstruction over the $K$-Draw Discrete Memoryless Symmetric Channel}
We consider sequence reconstruction over a discrete memoryless symmetric (DMS)
substitution channel with multiple observations per transmission, referred to as
\emph{$K$-draw}. A codeword $\bfx \in \cC \subseteq \Sigma_q^n$ is transmitted over a $q$-ary symmetric channel, denoted by $\mathsf{DMS}_q(p)$, where each symbol is independently substituted with probability $p$ by a symbol chosen uniformly from the remaining $q-1$ symbols. The receiver observes $K$ independent noisy copies (or \emph{reads}) of the same
codeword, $Y := \{\{ \bfy_0, \bfy_1, \ldots, \bfy_{K-1}\}\}$,  where each $\bfy_i$ is obtained by passing $\bfx$ through an independent instance of $\mathsf{DMS}_q(p)$. We refer to this channel model as the \emph{$K$-draw discrete memoryless symmetric channel} and denote it by $\mathsf{DMS}_{q,K}(p)$. The objective is to efficiently reconstruct the transmitted codeword from the multiset of noisy reads.



The capacity of the $\mathsf{DMS}_{2,K}(p)$ channel was derived in \cite{M06}. In the following theorem, whose proof can be found in the appendix, we extend this result to the general case of arbitrary $q$, i.e., we compute the capacity of the $\mathsf{DMS}_{q,K}(p)$ channel.
We define the $q$-ary entropy function as
\begin{align*}
H_q(p)
:= \frac{
p \log(q-1) - p\log p - (1-p)\log(1-p)
}{\log q}.
\end{align*}

\begin{restatable}{theorem}{DMSqKpCapacity}\label{thm:DMSqKpCapacity}
For $p \in (0,1)$ and $q,K \geq 2$, the capacity of the
$\mathsf{DMS}_{q,K}(p)$ channel is
\begin{align*}
C_{\mathsf{DMS}_{q,K}(p)} = \frac{p^K}{q\log q\,(q-1)^K}
\Biggl[
q\Bigl(\tf{q-1}{p}\Bigr)^K
\Bigl(
K\log\!\Bigl(\tf{q-1}{p}\Bigr)
+ \log q
\Bigr)
- \Phi
\Biggr] - K H_q(p),
\end{align*}
where
\vspace{-0.3cm}
\begin{align*}
\Phi
= \sum_{\substack{\sum_{i=0}^{q-1} r_i = K\\ r_i \ge 0}}
\binom{K}{r_0,\ldots,r_{q-1}}
\Biggl(
\sum_{i=0}^{q-1}
\Bigl(\tf{1-p}{p}\Bigr)^{r_i}
(q-1)^{r_i}
\Biggr)
\log\!\Biggl(
\sum_{i=0}^{q-1}
\Bigl(\tf{1-p}{p}\Bigr)^{r_i}
(q-1)^{r_i}
\Biggr).
\end{align*}
Moreover, $\lim_{q\to\infty}
C_{\mathsf{DMS}_{q,K}(p)}
= 1 - p^K$.
\end{restatable} 

\vspace{-0.3cm}
The channel capacity provides a natural benchmark for evaluating the achievable reconstruction rates derived later and is plotted in Figure~\ref{fig:rateRegion-comp}. 

\vspace{-0.1cm}
\section{RS Codes over the $\mathsf{DMS}_{q,K}(p)$ Channel}
\vspace{-0.1cm}
In this subsection, we design an efficient reconstruction algorithm over the $\mathsf{DMS}_{q,K}(p)$ tailored to RS codes. We assume that $ 2 \leq K  \ll q$ and $p \in (0,1)$.  In Algorithm~\ref{alg:mult-mat}, we outline the procedure for constructing the multiplicity matrix from the received set of reads. Then, in Algorithm~\ref{alg:read-rec-K-reads}, we present our reconstruction algorithm. For $n \geq 1$ and $i \in [n]$, let $\bfe_i \in \{0, 1\}^n$ denote the standard basis vector, which has a $1$ at the $i$-th position and $0$ elsewhere.

\vspace{0.2cm}
\begin{algorithm}[H]
    \caption{Multiplicity matrix constructor}
    \label{alg:mult-mat}
    \SetAlgoLined
    \DontPrintSemicolon
    
    \SetKwInOut{Input}{input}
    \SetKwInOut{Output}{output}

    \Input{A set of $K$ reads $Y$ and an integer $\mu$}
    \Output{A multiplicity matrix $M\in \FF_q^{q\times n}$}
    Set $M = \textbf{0}_{q\times n}$ \;
    \For{$j \in [n]$}{
        \For{$\bfy \in Y$}{
            Set $i$ such that $\delta_i = y_j$ and update $(M)_{i,j} = (M)_{i,j} + \mu$
        }
    }
    \For{$i \in [n]$}{
    Let $m_i = \max\{(M)_{:,i}\}$ be the maximum multiplicity at the $i$-th position\\
    Let $V_i = \argmax\{(M)_{:,i}\}$ be the corresponding set of field elements with the maximum multiplicity\\
        \If{$m_i > 1$ and $|V_i| = 1$}{
            Let $v_i \in V_i$ and set $(M)_{:,i} = \mu K \bfe_{v_i}$ \;
        }
    }
    Return $M$ 
\end{algorithm}

\begin{algorithm}[H]
    \SetAlgoLined
    \DontPrintSemicolon
    \SetKwInOut{Input}{input}
    \SetKwInOut{Output}{output}

    \Input{A set of $K$ reads $Y$, and integer $\mu$}
    \Output{A codeword $\bfc$}
    Generate a multiplicity matrix $M$ using Algorithm~\ref{alg:mult-mat} with input $Y$ and an integer $\mu$\;
    Run the KV algorithm with the multiplicity matrix $M$ to produce a list of candidates $\cL$\;
    Let $\widehat{\bfc} = \underset{\bfc'
    \in \cL}{\argmax}~\cS_{M}(\bfc')$ \;
    \If{$\widehat{\bfc}$ is not unique}
    {
    Return FAILURE
    }
    \Else{Return the codeword $\widehat{\bfc}$}
    \caption{Reconstruction Algorithm}
    \label{alg:read-rec-K-reads}
\end{algorithm}

\vspace{0.3cm}
Next, we analyze the performance of Algorithm~\ref{alg:read-rec-K-reads}. 
Let $\bfX \in \mathbb{F}_q^n$ denote the random variable corresponding to the channel input (over $n$ uses of the channel). Let $\widehat{\bfc}$ be the output of Algorithm~\ref{alg:read-rec-K-reads}. Assuming a uniform transmission probability over the code $\cC$, the probability of error is:
\begin{align*}
    \Perr &= \sum_{\bfc\in\cC}P(\widehat{\bfc} \neq \bfc \mid \bfX = \bfc) P_\bfX(\bfX = \bfc)\\
    &=  \frac{\sum_{\bfc\in\cC}P(\widehat{\bfc} \neq \bfc \mid \bfX = \bfc)}{q^k}\;.
\end{align*}
The probability $P(\widehat{\bfc} \neq \bfc \mid \bfX = \bfc)$ for any arbitrary $\bfc \in \cC$ satisfies: 
\begin{align*}
    &P(\widehat{\bfc} \neq \bfc \mid \bfX = \bfc) \\
    &= P(\bfc \not\in\cL) + P( \bfc \in \cL) \cdot P\left(\bfc \neq \underset{\bfc'
    \in \cL}{\argmax}~\cS_{M}(\bfc') \right)\\
    &\leq  1 - P(\bfc \in\cL) + P\left(\bfc \neq \underset{\bfc' \in \cL}{\argmax}~\cS_{M}(\bfc') \right)\\
    &\overset{(*)}{\leq} 1 \hspace{-0.2ex}-\hspace{-0.2ex} P\big(\cS_M(\bfc) \hspace{-0.2ex}\geq\hspace{-0.2ex}\sqrt{2(k-1)C(M)}\big) \hspace{-0.2ex}+\hspace{-0.2ex} P\big(\bfc \hspace{-0.2ex}\neq\hspace{-0.2ex} \underset{\bfc'
    \in \cL}{\argmax}~\cS_{M}(\bfc') \big),
\end{align*}
where $(*)$ follows from Theorem~\ref{thm:kv-alg}.

Therefore, we get the following upper bound on the probability of error: 
\begin{align}\label{Ineq:Perr_upBound}
\Perr
&\le \frac{1}{q^k}
\sum_{\bfc\in\cC}
\Bigl[
1
- P\!\Bigl(
\cS_M(\bfc)
\ge \sqrt{2(k-1)C(M)}
\Bigr)
+ P\!\Bigl(
\bfc \neq
\argmax_{\bfc'\in\cL}
\cS_M(\bfc')
\Bigr)
\Bigr].
\end{align}

Consequently, we derive lower and upper bounds on 
\[P\left(\cS_M(\bfc) \geq \sqrt{2(k-1)C(M)}\right)~\text{and}\;\;P\left(\bfc \neq \underset{\bfc'\in \cL}{\argmax}~\cS_{M}(\bfc') \right),\] 
respectively, to bound the probability of error $\Perr$. 

To facilitate our analysis, we define the following auxiliary random variables that characterize the possible outcomes at position  \( i \in [n] \) as the multiplicity matrix is updated in Algorithm~\ref{alg:mult-mat}: 
\begin{enumerate}
    \item \(\Eras_{i,A}\):
    This variable corresponds to the scenario when the transmitted symbol
    \( c_i \) appears exactly once and the majority is also one.
    \vspace{-0.15cm}
    \begin{align*}
        \Eras_{i,A} := \mathbb{I}\{ m_i = 1 \} \cdot \mathbb{I}\{ M_{c_i,i} = 1 \}\;.
    \end{align*}
    \item \(\Eras_{i,B}\):
    This variable corresponds to the scenario when the transmitted symbol
    \( c_i \) does not appear at all and the majority is one.
    \vspace{-0.15cm}
    \[
    \Eras_{i,B} := \mathbb{I}\{ m_i = 1 \} \cdot \mathbb{I}\{ M_{c_i,i} = 0 \}\;.
    \]

    \item \(\Suc_i\):
    This variable corresponds to the scenario when there is a unique majority, and that majority is achieved by the transmitted symbol \( c_i \). Recall that $V_i = \argmax\{(M)_{:,i}\}$ and $v_i \in V_i$. 
    \vspace{-0.15cm}
    \[
    \Suc_i := \mathbb{I}\{ |V_i| = 1 \}
    \cdot \mathbb{I}\{ m_i > 1 \}
    \cdot \mathbb{I}\{ v_i = c_i \}\;.
    \]
    \item \(\Err_i\):
    This variable captures all remaining scenarios, including incorrect majority
    or ties:
    \vspace{-0.15cm}
    \[
    \Err_i := 1 - \Eras_{i,A} - \Eras_{i,B} - \Suc_i\;.
    \]
\end{enumerate}
Note that each of the above random vectors is independent and identically distributed (i.i.d.) across $[n]$. Hence, in what follows, we drop the subscript \( i \) and consider an arbitrary index \( i \in [n] \), unless explicitly stated otherwise. Furthermore, the random variables \( \Err \), \( \Eras_A \), \( \Eras_B \), and \( \Suc \) are mutually exclusive. In the following lemma, we compute the probability mass functions of these random variables.  

\begin{lemma}
The probabilities that the random variables \( \Eras_A \), \( \Eras_B \), and \( \Suc \) take the value \( 1 \) are: 
\begin{align*}
    P(\Eras_{A} = 1) &= K! \cdot (1-p) \cdot \binom{q-1}{K-1} \left( \frac{p}{q-1} \right)^{K-1}\hspace{-0.55cm}, \\
    P(\Eras_{B} = 1) &= K! \cdot\binom{q-1}{K} \left( \frac{p}{q-1} \right)^{K}\hspace{-0.2cm}, \\
    P(\Suc = 1) &= \sum_{j = 2}^{K} \binom{K}{j} (1-p)^j \left( \frac{p}{q-1} \right)^{K-j}\hspace{-0.45cm}\sum_{\substack{ \sum_{i=0}^{q-2} r_i = K-j \\ j > r_0, r_1, \ldots, r_{q-2} \geq 0 }} 
    \binom{K-j}{r_0, r_1, \ldots, r_{q-2}} .
\end{align*}
\end{lemma}

\begin{proof}
The random variable $\Eras_{A}$ takes value $1$ when the majority count is $1$ and the correct symbol appears only once at that position. Thus,
\begin{align*}
    P(\Eras_{A} = 1 ) &= P(m_i = 1)\cdot P(M_{c_i,i} = 1) \\
    &= K! (1-p) \binom{q-1}{K-1} \left(\frac{p}{q-1}\right)^{K-1}\;.
\end{align*}

Similarly, the random variable $\Eras_{B}$ takes value $1$ when the majority count is $1$ and all the reads contain incorrect symbols at that position. Thus, 
\begin{align*}
    P(\Eras_{B} = 1) &= P(m_i = 1) \cdot P(M_{c_i,i} = 0) \\
    &= K! \binom{q-1}{K} \left(\frac{p}{q-1}\right)^{K}\;.
\end{align*}

Considering cases where $K \geq j \geq 2$ symbols match $c_i$:

\begin{align*}
    P(\Suc = 1) & = P(|V_i| = 1) \cdot P(m_i > 1)\cdot P(v_i = c_i)\\
    &= \sum_{j = 2}^{K} P(|V_i| = 1) \cdot P(M_{v_i,i} = j)\cdot P(v_i = c_i) \\
    &= \sum_{j = 2}^{K} \binom{K}{j} (1-p)^j \left(\frac{p}{q-1}\right)^{K-j}\hspace{-0.45cm}\sum_{\substack{\sum_{i=0}^{q-2}r_i=K-j \\ j > r_0,r_1,\ldots, r_{q-2} \geq 0}} 
    \binom{K-j}{r_0,r_1,\ldots, r_{q-2}}\;.
\end{align*}
\end{proof}

In the next lemma and corollary, whose proofs can be found in the appendix, we derive upper and lower bounds on these probabilities.

\begin{restatable}{lemma}{RVProbBounds}\label{lem:rv_prob_bounds}
The probabilities of the random variables \( \Suc \), \( \Eras_{A} \), \( \Eras_{B} \), and \( \Err \) taking the value \( 1 \) satisfy:
\begin{align*}
    &P(\Suc = 1)  \in \left(\cL_S, \cU_S\right), \\ 
    &P(\Eras_{A} = 1)  \in \left(\cL_{E_A},\cU_{E_A}\right), \\
    &P(\Eras_{B} = 1)  \in \left(\cL_{E_B},\cU_{E_B}\right)\\
    &P(\Err = 1) \in\left(0, \cU_{E}\right).
\end{align*}
where
\begin{align*}
    \cL_S &:= \left(1 - \frac{p(K-1)}{q-1}\right) 
    - \left(\frac{p(q-K)}{q-1}\right)^{K} - K(1-p)\left(\frac{p(q-K)}{q-1}\right)^{K-1},~ 
    \cU_S := 1 - p^K - K(1-p)p^{K-1}\;\hspace{-0.4cm}, \\
    \cL_{E_A} &:= K(1-p)\left(\frac{p(q-K)}{q-1}\right)^{K-1},~
    \cU_{E_A}:= K(1-p)p^{K-1},\\
    \cL_{E_B}&:= \left(\frac{p(q-K)}{q-1}\right)^{K},~\cU_{E_B}:= p^K,\\
    \cU_{E}&:=\frac{p(K-1)}{q-1}\;.
\end{align*}
\end{restatable}

\begin{restatable}{corollary}{RVProbBoundsEps}\label{cor:rv_prob_bounds_eps}
Let \(\epsilon_q=\frac{K(K-1)}{q-1}\). Then, we have that:
\begin{align*}
    \cL_S &\geq 1 - p^K - K(1-p)p^{K-1} - \frac{\epsilon_q p}{K}, \\
     \cU_S &\leq 1 - p^K - K(1-p)p^{K-1}, \\
    \cL_{E_A} &\geq (1 - \epsilon_q) \cdot K(1-p)p^{K-1},~\cU_{E_A} \leq K(1-p)p^{K-1}, \\
\cL_{E_B} &\geq (1 - \epsilon_q) \cdot p^K,~\cU_{E_B} \leq p^K,~\cU_E \leq \frac{\epsilon_q p}{K}\;.
\end{align*}
\end{restatable}

Next, we define the following aggregate random variables, which count the number of positions corresponding to each outcome as the multiplicity matrix is updated in Algorithm~\ref{alg:mult-mat}: 
\begin{enumerate}
    \item $\NErr := \sum_{i=1}^n \Err_i$,
    \item $\NSuc := \sum_{i=1}^n \Suc_i$,
    \item $\NEras_A := \sum_{i=1}^n \Eras_{i,A}$,
    \item $\NEras_B := \sum_{i=1}^n \Eras_{i,B}$.
\end{enumerate}

In the following lemma, by applying concentration inequalities, we obtain a lower bound on the probability that all of the aggregate random variables simultaneously concentrate around their respective expectations.

\begin{lemma}
\label{lem:concentration_bounds}
Let $\epsilon_q = \frac{K(K-1)}{q-1}$. For $\eta \in (0,1)$, let $\delta_n = \sqrt{\frac{n \log(8/\eta)}{2}}$ and $\Psi_{\delta_n}$ denote the event that the following events hold simultaneously:
\begin{align*}
\NErr
&\in
\left(
0,\;
n\,\frac{\epsilon_q p}{K}
+ \delta_n
\right),
\\
\NSuc
&\in
\Bigl(
n\bigl(
1 - p^K - K(1-p)p^{K-1}
- \tfrac{\epsilon_q p}{K}
\bigr)
- \delta_n,
n\bigl(
1 - p^K - K(1-p)p^{K-1}
\bigr)
+ \delta_n
\Bigr),
\\
\NEras_A
&\in
\Bigl(
n(1-\epsilon_q)K(1-p)p^{K-1}
- \delta_n,
nK(1-p)p^{K-1}
+ \delta_n
\Bigr),
\\
\NEras_B
&\in
\Bigl(
n(1-\epsilon_q)p^K
- \delta_n,~
np^K + \delta_n
\Bigr)\;.
\end{align*}
Then, we have that 
 $P(\Psi_{\delta_n}) > 1-\eta\;$.
\end{lemma}
\begin{proof}
Note that the random variables $\{\Err_i\}, \{\Suc_i\}, \{\Eras_{i,A}\}, \{\Eras_{i,B}\}$ are i.i.d. Bernoulli random variables. Hence, the aggregate variables $\NErr$, $\NSuc$, $\NEras$, $\NEras_A$, and $\NEras_B$ are sums of $n$ independent bounded random variables in $[0,1]$.

By Hoeffding’s inequality, for a sum $X = \sum_{i=1}^n X_i$ of independent variables $X_i \in [0,1]$:
\[
P\left( \left| X - \mathbb{E}[X] \right| \geq \delta \right) \leq 2 \exp\left( -\frac{2 \delta^2}{n} \right).
\]
Setting
\[
\delta = \delta_n := \sqrt{\frac{n \log(8/\eta)}{2}},
\]
we obtain:
\[
P\left( \left| X - \mathbb{E}[X] \right| \geq \delta_n \right) \leq \frac{\eta}{4},
\]
or equivalently,
\begin{align}\label{ineq:hoeff_sum}
P\left(X \in \left[\mathbb{E}[X] - \delta_n, \mathbb{E}[X] + \delta_n\right]\right) \geq 1 - \frac{\eta}{4}.
\end{align}

We now compute the expectations and apply the inequality to each variable. From Corollary~\ref{cor:rv_prob_bounds_eps}, we have:
\begin{align*}
\mathbb{E}[\NErr] &< n \cdot \frac{\epsilon_q p}{K}\;,\\
\mathbb{E}[\NSuc] &\in \left[
n \left(1 - p^K - K(1-p)p^{K-1} - \frac{\epsilon_q p}{K} \right),
n \left(1 - p^K - K(1-p)p^{K-1} \right)
\right]\;,\\
\mathbb{E}[\NEras_A] &\in \left[
n(1 - \epsilon_q)K(1-p)p^{K-1},
nK(1-p)p^{K-1}
\right]\;,\\
\mathbb{E}[\NEras_B] &\in \left[
n(1 - \epsilon_q)p^K,
np^K
\right],
\end{align*}
Thus, with probability at least $1 - \frac{\eta}{4}$, each of the following inequalities hold:
\begin{align*}
\NErr &\in \left(0, n \cdot \frac{\epsilon_q p}{K} + \delta_n\right)\;,\\
\NSuc &\in \left(
n \left(1 - p^K - K(1-p)p^{K-1} - \frac{\epsilon_q p}{K} \right) - \delta_n,
n \left(1 - p^K - K(1-p)p^{K-1} \right) + \delta_n\right)\;,\\
\NEras_A &\in \left(n(1 - \epsilon_q)K(1-p)p^{K-1} - \delta_n,
nK(1-p)p^{K-1} + \delta_n
\right)\;,\\
\NEras_B &\in \left(
n(1 - \epsilon_q)p^K - \delta_n,
np^K + \delta_n
\right)\;.
\end{align*}
Using the union bound over the four events, the probability that all four events hold simultaneously is at least $1- 4 \cdot \frac{\eta}{4} = 1- \eta$, as desired.  
\end{proof}

Therefore, by Lemma~\ref{lem:concentration_bounds}, the random variables
$\NEras_A$ and $\NEras_B$ concentrate with high probability around
$nK(1-p)p^{K-1}$ and $np^K$, respectively.
Since a hard-decision majority-vote decoder declares these coordinates as erasures,
we obtain the following corollary.

\begin{corollary}\label{cor:majority}
Let $p \in (0,1)$ and $K \ge 2$.
Then, for sufficiently large blocklength, the transmitted codeword can be reconstructed with arbitrarily small probability of error using the hard-decision majority-vote decoder whenever

\vspace{-1.5ex}
\begin{small}
\[
R \le 1 - p^K - K(1-p)p^{K-1}.
\]
\end{small}
\end{corollary}

Next, we analyze Algorithm~\ref{alg:read-rec-K-reads}, which applies the Koetter–Vardy (KV) soft-decoding algorithm using the multiplicity matrix produced by Algorithm~\ref{alg:mult-mat}. In the following theorem, we quantify the contribution of the aggregate random variables in order to bound both the cost of the resulting multiplicity matrix and the score of candidate codewords. These bounds are then used to derive sufficient conditions on the parameters $(n,k,q,K,\mu,p)$ under which the probability of error $\Perr$ can be made arbitrarily small.

\begin{theorem}\label{thm:Perr_UpBound}
Let $K \geq 2$, $p \in (0,1)$ and $\mu \in \mathbb{N}$. Define:
\begin{align*}
\epsilon_q
&:= \frac{K(K-1)}{q-1},~
\delta_n := \sqrt{\frac{n\log(8/\eta)}{2}}, \alpha:= p^K + K(1-p)p^{K-1},
\\
O_n^*
&:= \mu\Bigl(
n\bigl(
\epsilon_q p + \alpha
\bigr)
+ (k-1)K(1-\alpha)
\nonumber
+ \delta_n\Bigl(
K + 2 + \tfrac{(k-1)K}{n}
\Bigr)
\Bigr),\\
S_n^*
&:= \mu\Bigl(
n\Bigl(
K(1-\alpha)
- \epsilon_q p
+ K(1-\epsilon_q)(1-p)p^{K-1}
\Bigr) 
\hspace{-0.3ex}- \hspace{-0.3ex}(K+1)\delta_n
\Bigr),\\
C_n^*
&:= \frac{\mu K}{2}
\Bigl(
n\Bigl(
(\mu K-1)\bigl(1-(1-\epsilon_q)\alpha\bigr)
+ (\mu-1)\alpha
\Bigr)
+ 2\delta_n\bigl(\mu(K+1)-2\bigr)
\Bigr)\;.
\end{align*}
Then, for any arbitrary small $\eta \in (0,1)$, 
\begin{align*}
    \Perr \leq 2\eta\;,
\end{align*}
provided that $S_n^* \geq \max\left\{\sqrt{2(k-1)C_n^*}, O_n^*\right\}$.
\end{theorem}
\begin{proof} 
The cost of the multiplicity matrix $M$ satisfies
\begin{small}
\begin{align*}
    C(M) &\leq (n-\NEras_A-\NEras_B) \cdot\binom{\mu K}{2} + \left(\NEras_A+\NEras_B\right) \cdot K\binom{\mu}{2}\;.
\end{align*}
\end{small}
The score of the transmitted codeword $\bfc$ with respect to $M$ is lower bounded as
\begin{small}
\begin{align*}
    \cS_M(\bfc) \geq \NSuc\cdot\mu K + \NEras_A \cdot \mu\;. 
\end{align*}
\end{small}
Consequently, 
\begin{align*}
&P\left( \cS_M(\bfc) \geq \sqrt{2(k-1)C(M)} \right) \\
&\quad> P\left( \mathcal{S}(\bfc) \geq \sqrt{2(k-1)C(M)}~\bigg\vert~\Psi_{\delta_n}\right)\cdot P\left(\Psi_{\delta_n}\right)\\
&\quad\overset{(1)}{>} \mathbb{I}\left\{ S_n^* \geq  \sqrt{2(k-1)C_n^*} \right\} \cdot P\left(\Psi_{\delta_n}\right)\\
&\quad> \mathbb{I}\left\{ S_n^* \geq  \sqrt{2(k-1)C_n^*} \right\} \cdot (1-\eta)\;,
\end{align*}
where $(1)$ follows since, conditioned on $\Psi_{\delta_n}$,
$\cS_M(\bfc) \ge S_n^*$ and $C(M) \le C_n^*$. 

Next, the maximum score among all but the transmitted codeword in the list generated by Algorithm~\ref{alg:read-rec-K-reads} is at most 
\begin{align*}
\underset{\bfc'
    \in \cL\backslash\{\bfc\}}{\max}~\cS_{M}(\bfc') &\leq \NErr\cdot\mu K + \frac{k-1}{n}\NSuc\cdot\mu K + (\NEras_A + \NEras_B) \cdot \mu\;.
\end{align*}
Then, 
\begin{align*}
    &P\left(\bfc \neq \underset{\bfc'
    \in \cL}{\argmax}~\cS_{M}(\bfc')\right)  \\
    &< P\left(\bfc \neq \underset{\bfc'
    \in \cL}{\argmax}~\cS_{M}(\bfc')~\bigg\vert~\Psi_{\delta_n}\right)\cdot P\left(\Psi_{\delta_n}\right) + 1 - P\left(\Psi_{\delta_n}\right) \\
    &< P\left( \underset{\bfc'
    \in \cL\backslash\{\bfc\}}{\max}~\cS_{M}(\bfc') > \mathcal{S}_M(\bfc)~\bigg\vert~\Psi_{\delta_n}\right) + \eta \\
    &\overset{(2)}{<} \mathbb{I}\left\{ O_n^* > S_n^* \right\} + \eta\;,
\end{align*}
where $(2)$ follows since, conditioned on $\Psi_{\delta_n}$,
$\cS_M(\bfc) \ge S_n^*$ and
$\max_{\bfc'\in\cL\setminus\{\bfc\}} \cS_M(\bfc') \le O_n^*$.

Finally, since $S_n^* \geq \max\left\{\sqrt{2(k-1)C_n^*}, O_n^*\right\}$, from Inequality~(\ref{Ineq:Perr_upBound}), it follows that:
\begin{align*}
\Perr
&\le
\frac{1}{q^k}
\sum_{\bfc\in\cC}
\Bigl[
1
- P\!\Bigl(
\cS(\bfc)
\ge \sqrt{2(k-1)C(M)}
\Bigr) 
+ P\!\Bigl(
\bfc \neq
\argmax_{\bfc'\in\cL}
\cS_M(\bfc')
\Bigr)
\Bigr]
\le 2\eta,
\end{align*}
as desired.
\end{proof}

A direct simplification of Theorem~\ref{thm:Perr_UpBound} leads to the following corollary, which characterizes the asymptotic behavior of the quantities appearing in Theorem~\ref{thm:Perr_UpBound} and derives a threshold on the coding rate.

\begin{corollary}\label{cor:asymp_values}
Let $p, \eta\in(0,1)$ and $K \geq 2$ and $\mu\sim\cO(1)$. Let $\alpha := p^K + K(1-p)\,p^{K-1}$. Then, 
\begin{align*}
O_n^* 
&\sim 
(\mu 
\alpha) n +  \cO\left(\sqrt{n}\right)\\
S_n^* 
&\sim 
\mu
K \Bigl(1-p^K - (K-1)(1-p)p^{K-1}
\Bigr)n - \cO\left(\sqrt{n}\right)\;, \\
C_n^* &\sim 
\frac{\mu\,K}{2}\;
\Bigl((\mu K -1)(1 -\alpha) + (\mu-1)\,\alpha
\Bigr) n +  \cO\left(\sqrt{n}\right)
\;.
\end{align*}
Furthermore, if 
\begin{align*}
R &\le
\frac{
K\bigl(1 - p^K - (K-1)(1-p)p^{K-1}\bigr)^2
}{
1 + (K-1)\bigl(1 - p^K - K(1-p)p^{K-1}\bigr)
}, \\
p &\le K^{-1/(K-1)},
\end{align*}
then it holds that
\begin{align*}
    \underset{n \rightarrow \infty}{\lim} S_n^* \geq \underset{n \rightarrow \infty}{\lim} \max\left\{\sqrt{2(k-1)C_n^*}, O_n^*\right\}\;.
\end{align*}
\end{corollary}
Upon combining Theorem~\ref{thm:Perr_UpBound} and Corollary~\ref{cor:asymp_values}, we obtain the following result.
    \begin{theorem}\label{thm:R_upBound_DMS}
Let $p \in (0,1)$ and $K \ge 2$.
Then, for sufficiently large blocklength and integer $\mu$, the transmitted codeword can be
reconstructed in polynomial-time with arbitrarily small probability of error using
Algorithm~\ref{alg:read-rec-K-reads} provided that

\vspace{-2ex}
\begin{small}
\begin{align*}
R &\le
\frac{
K\bigl(1 - p^K - (K-1)(1-p)p^{K-1}\bigr)^2
}{
1 + (K-1)\bigl(1 - p^K - K(1-p)p^{K-1}\bigr)
}, \\
p &\le K^{-1/(K-1)}.
\end{align*}
\end{small}
\end{theorem}

\section{Conclusion}

In this work, we studied the problem of reconstructing a transmitted RS codeword from multiple independent noisy reads over a $q$-ary DMS substitution channel. By adapting the Koetter--Vardy soft-decision decoding framework to the $K$-draw setting, we developed a polynomial-time reconstruction algorithm that aggregates reliability information across reads through a carefully designed multiplicity matrix. For sufficiently large blocklength and alphabet size, we derived an explicit rate threshold depending only on the channel parameters $(p,K)$, below which the reconstruction algorithm can recover the transmitted codeword with arbitrarily small probability of error.

\newpage

\bibliographystyle{plain}
\bibliography{refJournalBibTex}

@article{AY21,
  author = {M. Abu-Sini and E. Yaakobi},
  title = {On Levenshtein’s reconstruction problem under insertions, deletions, and substitutions},
  journal = {IEEE Trans on Inf. Theory},
  volume = {67},
  number = {11},
  pages = {7132--7158},
  year = {2021}
}

@article{CGK12,
  author = {G. M. Church and Y. Gao and S. Kosuri},
  title = {Next-generation digital information storage in DNA},
  journal = {Science},
  volume = {337},
  number = {6102},
  pages = {1628--1628},
  year = {2012},
  month = {Sep.}
}

@inproceedings{GY16,
  author = {R. Gabrys and E. Yaakobi},
  title = {Sequence reconstruction over the deletion channel},
  booktitle = {Proc. Int. Symp. on Inf. Theory},
  pages = {1596--1600},
  year = {2016},
  month = {Jul.}
}

@article{Getal13,
  author = {N. Goldman and others},
  title = {Towards practical, high-capacity, low-maintenance information storage in synthesized DNA},
  journal = {Nature},
  volume = {494},
  number = {7435},
  pages = {77--80},
  year = {2013}
}

@article{GS99,
  author = {V. Guruswami and M. Sudan},
  title = {Improved decoding of Reed-Solomon and algebraic-geometry codes},
  journal = {IEEE Trans on Inf. Theory},
  volume = {45},
  number = {6},
  pages = {1757--1767},
  year = {1999},
  month = {Sept.}
}

@article{JV04,
  author = {T. Jiang and A. Vardy},
  title = {Asymptotic improvement of the Gilbert-Varshamov bound on the size of binary codes},
  journal = {IEEE Trans on Inf. Theory},
  volume = {50},
  number = {8},
  pages = {1655--1664},
  year = {2004},
  month = {Aug.}
}

@inproceedings{JLL23,
  author = {V. Junnila and T. Laihonen and T. Lehtilä},
  title = {Levenshtein’s Reconstruction Problem with Different Error Patterns},
  booktitle = {IEEE Int. Symp. on Inf. Theory},
  pages = {1300--1305},
  year = {2023},
  address = {Taipei, Taiwan}
}

@article{JL14,
  author = {V. Junnila and T. Laihonen},
  title = {Codes for information retrieval with small uncertainty},
  journal = {IEEE Trans. on Inf. Theory},
  volume = {60},
  number = {2},
  pages = {976--985},
  year = {2014},
  month = {Feb.}
}

@article{JL15,
  author = {V. Junnila and T. Laihonen},
  title = {Information retrieval with unambiguous output},
  journal = {Inf. and Computation},
  volume = {242},
  pages = {354--368},
  year = {2015},
  month = {Jun.}
}

@article{JL16,
  author = {V. Junnila and T. Laihonen},
  title = {Information retrieval with varying number of input clues},
  journal = {IEEE Trans. on Inf. Theory},
  volume = {62},
  number = {2},
  pages = {625--638},
  year = {2016},
  month = {Feb.}
}

@article{KV03,
  author = {R. Koetter and A. Vardy},
  title = {Algebraic soft-decision decoding of Reed-Solomon codes},
  journal = {IEEE Trans on Inf. Theory},
  volume = {49},
  number = {11},
  pages = {2809--2825},
  year = {2003},
  month = {Nov.}
}

@article{K08,
  author = {E. Konstantinova},
  title = {On reconstruction of signed permutations distorted by reversal errors},
  journal = {Discrete Math.},
  volume = {308},
  number = {5-6},
  pages = {974--984},
  year = {2008},
  month = {Mar.}
}

@article{K07,
  author = {E. Konstantinova},
  title = {Reconstruction of permutations distorted by reversal errors},
  journal = {Discrete Applied Math.},
  volume = {155},
  number = {18},
  pages = {2426--2434},
  year = {2007}
}

@misc{KLS07,
  author = {E. Konstantinova and V. I. Levenshtein and J. Siemons},
  title = {Reconstruction of permutations distorted by single transposition errors},
  year = {2007},
  note = {arXiv:0702.191v1}
}

@article{LKKM08,
  author = {V. Levenshtein and E. Konstantinova and E. Konstantinov and S. Molodtsov},
  title = {Reconstruction of a graph from 2-vicinities of its vertices},
  journal = {Discrete Appl. Math.},
  volume = {156},
  number = {9},
  pages = {1399--1406},
  year = {2008},
  month = {May}
}

@article{LS09,
  author = {V. I. Levenshtein and J. Siemons},
  title = {Error graphs and the reconstruction of elements in groups},
  journal = {J. Comb. Theory Ser. A},
  volume = {116},
  number = {4},
  pages = {795--815},
  year = {2009},
  month = {May}
}

@article{L01A,
  author = {V. I. Levenshtein},
  title = {Efficient reconstruction of sequences},
  journal = {IEEE Trans. Inf. Theory},
  volume = {47},
  number = {1},
  pages = {2--22},
  year = {2001},
  month = {Jan.}
}

@article{L01B,
  author = {V. I. Levenshtein},
  title = {Efficient reconstruction of sequences from their subsequences or supersequences},
  journal = {J. Comb. Theory Ser. A},
  volume = {93},
  number = {2},
  pages = {310--332},
  year = {2001},
  month = {Feb.}
}

@inproceedings{M06,
  author = {M. Mitzenmacher},
  title = {On the Theory and Practice of Data Recovery with Multiple Versions},
  booktitle = {IEEE Int. Symp. on Inf. Theory},
  pages = {982--986},
  year = {2006},
  address = {Seattle, WA, USA}
}

@inproceedings{PGK22,
  author = {V.~L.~P. Pham and K. Goyal and H.~M. Kiah},
  title = {Sequence Reconstruction Problem for Deletion Channels: A Complete Asymptotic Solution},
  booktitle = {IEEE Int. Symp. on Inf. Theory},
  pages = {992--997},
  year = {2022},
  address = {Espoo, Finland}
}

@article{RS60,
  author = {I. S. Reed and G. Solomon},
  title = {Polynomial codes over certain finite fields},
  journal = {Journal of the Society for Industrial and Applied Mathematics},
  volume = {8},
  number = {2},
  pages = {300--304},
  year = {1960}
}

@article{SGSD17,
  author = {F. Sala and R. Gabrys and C. Schoeny and L. Dolecek},
  title = {Exact reconstruction from insertions in synchronization codes},
  journal = {IEEE Trans on Inf. Theory},
  volume = {63},
  number = {4},
  pages = {2428--2445},
  year = {2017}
}

@INPROCEEDINGS{SCKY24,
  author={S. Singhvi and R. Con and H.M. Kiah and E. Yaakobi},
  booktitle={2024 IEEE Int. Symp. on Inf. Theory (ISIT)}, 
  title={An Optimal Sequence Reconstruction Algorithm for Reed-Solomon Codes}, 
  year={2024},
  volume={},
  number={},
  pages={2832-2837}}

@article{YB18,
  author = {E. Yaakobi and J. Bruck},
  title = {On the uncertainty of information retrieval in associative memories},
  journal = {IEEE Trans. Inform. Theory},
  volume = {65},
  pages = {2155--2165},
  year = {2018}
}

@inproceedings{YSLB13,
  author = {E. Yaakobi and M. Schwartz and M. Langberg and J. Bruck},
  title = {Sequence reconstruction for grassmann graphs and permutations},
  booktitle = {IEEE Int. Symp. on Inf. Theory},
  pages = {874--878},
  year = {2013},
  month = {July}
}

@article{YKGMZM15,
  author = {S.H.T. Yazdi and H.M. Kiah and E. Garcia-Ruiz and J. Ma and H. Zhao and O. Milenkovic},
  title = {DNA-based storage: Trends and methods},
  journal = {IEEE Trans. Mol., Bio. and Multi-Scale Com.},
  volume = {1},
  pages = {230--248},
  year = {2015}
}

@article{Grass2015ChemicalPreservation,
  author  = {R. N. Grass and R. Heckel and M. Puddu and D. Paunescu and W. J. Stark},
  title   = {Robust chemical preservation of digital information on {DNA} in silica with error-correcting codes},
  journal = {Angew. Chem. Int. Ed.},
  volume  = {54},
  number  = {8},
  pages   = {2552--2555},
  year    = {2015}
}

@article{Erlich2017DNAFountain,
  author  = {Y. Erlich and D. Zielinski},
  title   = {{DNA} fountain enables a robust and efficient storage architecture},
  journal = {Science},
  volume  = {355},
  number  = {6328},
  pages   = {950--954},
  year    = {2017}
}

@article{Organick2018RandomAccess,
  author  = {L. Organick and S. D. Ang and Y.-J. Chen and R. Lopez and S. Yekhanin and K. Makarychev and M. Z. R{\'a}cz and G. Kamath and B. Nguyen and C. N. Takahashi and S. Newman and H.-Y. Parker and C. Rashtchian and K. Stewart and G. Gupta and R. Carlson and J. Mulligan and D. Carmean and G. Seelig and L. Ceze and K. Strauss},
  title   = {Random access in large-scale {DNA} data storage},
  journal = {Nat. Biotechnol.},
  volume  = {36},
  number  = {3},
  pages   = {242--248},
  year    = {2018}
}

@article{Press2020HEDGES,
  author  = {W. H. Press and J. A. Hawkins and S. K. Jones, Jr. and J. M. Schaub and I. J. Finkelstein},
  title   = {{HEDGES} error-correcting code for {DNA} storage corrects indels and allows sequence constraints},
  journal = {Proc. Natl. Acad. Sci. U.S.A.},
  volume  = {117},
  number  = {31},
  pages   = {18489--18496},
  year    = {2020}
}

@inproceedings{WangYaakobiZhang2024SimpleConditions,
  author    = {C. Wang and E. Yaakobi and Y. Zhang},
  title     = {How to find simple conditions for successful sequence reconstruction?},
  booktitle = {Proc. IEEE Inf. Theory Workshop (ITW)},
  pages     = {627--632},
  year      = {2024}
}

@inproceedings{PapadopoulouRameshwarWachterZeh2024Views,
  author    = {V. Papadopoulou and V. A. Rameshwar and A. Wachter-Zeh},
  title     = {On the expected number of views required for fixed-error sequence reconstruction},
  booktitle = {Proc. IEEE Inf. Theory Workshop (ITW)},
  pages     = {639--644},
  year      = {2024}
}

@misc{pham25,
  title={Sequence reconstruction for sticky insertion/deletion channels},
  author={Pham, Van Long Phuoc and Chee, Yeow Meng and Cai, Kui and Vu, Van Khu},
  year={2025},
  eprint={2504.19363},
  archivePrefix={arXiv},
  primaryClass={cs.IT}
}

@misc{lan25,
  title={Sequence reconstruction under channels with multiple bursts of insertions or deletions},
  author={Lan, Zhaojun and Sun, Yubo and Yu, Wenjun and Ge, Gennian},
  year={2025},
  eprint={2504.20460},
  archivePrefix={arXiv},
  primaryClass={cs.IT}
}

@article{sun23,
  title={Sequence reconstruction under single-burst-insertion/deletion/edit channel},
  author={Sun, Yubo and Xi, Yiqing and Ge, Gennian},
  journal={IEEE Transactions on Information Theory},
  volume={69},
  number={7},
  pages={4466--4483},
  year={2023},
  publisher={IEEE}
}

\newpage

\section*{Appendix}

\DMSqKpCapacity*

\begin{proof}
Let $\bfX \in \mathbb{F}_q$ and $\bfY \in \mathbb{F}_q^K$ denote the random variables corresponding to the channel input and output, respectively. Then, the capacity of the $\mathsf{DMS}_{q,K}(p)$ channel is given by:
\[
C_{\mathsf{DMS}_{q,K}(p)} = \underset{P_\bfX(x)}{\max} \frac{H(\bfY) - H(\bfY|\bfX)}{\log(q)}\;.
\]
Since the $\mathsf{DMS}_{q,K}(p)$ channel is symmetric, the capacity is achieved by a uniform input distribution:
\[
C_{\mathsf{DMS}_{q,K}(p)} = \frac{H(\bfY) - H(\bfY|\bfX)}{\log(q)} \Bigg\vert_{P_\bfX(x) \sim \text{Unif}(q)}\;.
\]

We first compute the conditional entropy \(H(\bfY | \bfX)\): 

For $x \in \mathbb{F}_q$ and $r \in [K]$, let  
\[
Q_{r,x} := \left\{\bfy = (y_0, y_1, \ldots, y_{K-1}) \in \mathbb{F}_q^K : \sum_{i=0}^{K-1} \mathbb{I}\left\{y_i = x\right\} = r\right\}.
\]
This set contains all sequences where the symbol $x$ appears exactly $r$ times. The number of such sequences is:  
\[
\vert Q_{r,x} \vert = \binom{K}{r} (q-1)^{K-r}.
\]

The probability of observing $\bfy \in Q_{r,x}$ given that the input is $x$ is:
\[
P(\bfY = \bfy \mid \bfX = x) = (1-p)^{r} \left( \frac{p}{q-1} \right)^{K-r}.
\]

With this, we evaluate the conditional entropy $H(\bfY|\bfX)$ for the uniform input distribution.
\[
    H(\bfY \mid \bfX) = \sum_{x \in \mathbb{F}_q} P_{\bfX}(x) H(\bfY \mid \bfX = x)\;.
\]
Due to symmetry and uniform distribution of the input,
\[
    H(\bfY \mid \bfX) = H(\bfY \mid \bfX = 0)\;.
\]

Next, we evaluate the conditional entropy for $\bfX = 0$. 
\begin{align*}
   H(\bfY \mid \bfX = 0) &= \sum_{r=0}^{K} \binom{K}{r} (q-1)^{K-r} (1-p)^r \frac{p^{K-r}}{(q-1)^{K-r}} \log \left( \frac{(q-1)^{K-r}}{(1-p)^r p^{K-r}} \right)\\
    &= p^K \left[ \sum_{r=0}^{K} \binom{K}{r} \left( \frac{1-p}{p} \right)^r \left[ (K-r) \log \left( \frac{q-1}{p} \right) - r \log (1-p) \right] \right]\\
    &= p^K \left[ \log \left( \frac{q-1}{p} \right) \left( \sum_{r=0}^{K} (K-r) \binom{K}{r} \left( \frac{1-p}{p} \right)^r \right) - \log (1-p) \left( \sum_{r=0}^{K} r \binom{K}{r} \left( \frac{1-p}{p} \right)^r \right) \right]\;.
\end{align*}
Using binomial theorem separately on the two sums, we get that
\begin{align*}
    H(\bfY \mid \bfX = 0) &= p^K \left[ \log \left( \frac{q-1}{p} \right) \left( K \left( \frac{1}{p} \right)^{K-1} \right) - \log (1-p) \left( \left( \frac{1-p}{p} \right) K \left( \frac{1}{p} \right)^{K-1} \right) \right]\\
    &= K \left( H_2(p) + p \log (q-1) \right)\\
    &= K H_q(p) \log q\;.
\end{align*}

We now compute the entropy \(H(\bfY)\): 

For $i \in \{0, 1, \ldots, q-1\}$ and $r_i \in \mathbb{Z}_{\geq 0}$ such that $\sum_{i=0}^{q-1} r_i = K$, define:
\[
Q_{r_0, r_1, \ldots, r_{q-1}} = \left\{\bfy : \sum_{j=0}^{K-1} \mathbb{I}\left\{y_j = \mathbb{F}_q[i]\right\} = r_i, \forall i \right\}\;.
\]
This set contains all sequences $\bfy$ where for each $i \in \{0, 1, \ldots, q-1\}$, there are exactly $r_i$ occurrences of the symbol $\mathbb{F}_q[i]$ in the vector $\bfy$. Using a standard counting argument, the total number of possible sequences $\bfy$ in this set is the multinomial coefficient:
\[
\vert Q_{r_0, r_1, r_2, \ldots, r_{q-1}}\vert = \binom{K}{r_0, r_1, r_2, \ldots, r_{q-1}}\;.
\]

We now evaluate the probability of observing $\bfy \in Q_{r_0, r_1, r_2, \ldots r_{q-1}}$ for the uniform input distribution.
\begin{align*}
P(\bfY=\bfy)&=\sum_{x\in\Fq} P(\bfY=\bfy~\vert~\bfX=x) P(x) \\
& =\frac{1}{q} \sum_{x\in\Fq} P(\bfY=\bfy~\vert~\bfX=x)\\
&=\frac{1}{q} \sum_{i=0}^{q-1}(1-p)^{r_{i}}\left(\frac{p}{q-1}\right)^{K-r_{i}} \\
& =\frac{p^{K}}{q(q-1)^{K}} \sum_{i=0}^{q-1}\left(\frac{1-p}{p}\right)^{r_{i}}\left(q-1\right)^{r_{i}}\;.
\end{align*}
With this, we evaluate the entropy $H(\bfY)$ for the uniform input distribution as follows: 
\begin{align*}
    &H(\bfY) \\
    &= -\sum_{\bfy\in\Fq^K}P_Y(\bfy)\log(P_Y(\bfy))\\
     &= -\sum_{\substack{\sum_{i=0}^{q-1}r_i=K \\ r_0,r_1,\ldots,r_{q-1} \geq 0}}\sum_{\bfy\in Q_{r_0, r_1, \ldots, r_{q-1}}}P_Y(\bfy)\log(P_Y(\bfy))\\
    &= -\sum_{\substack{\sum_{i=0}^{q-1}r_i=K \\ r_0,r_1,\ldots,r_{q-1} \geq 0}} \binom{K}{r_0,r_1,\ldots r_{q-1}} \left(\frac{p^{K}}{q(q-1)^{K}} \sum_{i=0}^{q-1}\left(\frac{1-p}{p}\right)^{r_{i}}\left(q-1\right)^{r_{i}}\right)\log\left(\frac{p^{K}}{q(q-1)^{K}} \sum_{i=0}^{q-1}\left(\frac{1-p}{p}\right)^{r_{i}}\left(q-1\right)^{r_{i}}\right)\\
    &= \frac{p^{K}}{q(q-1)^{K}}\sum_{\substack{\sum_{i=0}^{q-1}r_i=K \\ r_0,r_1,\ldots,r_{q-1} \geq 0}} \binom{K}{r_0,r_1,\ldots r_{q-1}} \left( \sum_{i=0}^{q-1}\left(\frac{1-p}{p}\right)^{r_{i}}\left(q-1\right)^{r_{i}}\right)\log\left(\frac{q(q-1)^{K}}{p^{K}} \right)\\
    &~~-\frac{p^{K}}{q(q-1)^{K}}\sum_{\substack{\sum_{i=0}^{q-1}r_i=K \\ r_0,r_1,\ldots,r_{q-1} \geq 0}} \binom{K}{r_0,r_1,\ldots r_{q-1}} \left( \sum_{i=0}^{q-1}\left(\frac{1-p}{p}\right)^{r_{i}}\left(q-1\right)^{r_{i}}\right) \log\left(\sum_{i=0}^{q-1}\left(\frac{1-p}{p}\right)^{r_{i}}\left(q-1\right)^{r_{i}}\right)\\
    &\overset{(1)}{=} \frac{p^{K}}{q(q-1)^{K}}\log\left(\frac{q(q-1)^{K}}{p^{K}} \right)q\left((q-1) + (q-1)\frac{1-p}{p}\right)^K\\
    &~~-\frac{p^{K}}{q(q-1)^{K}}\sum_{\substack{\sum_{i=0}^{q-1}r_i=K \\ r_0,r_1,\ldots,r_{q-1} \geq 0}} \binom{K}{r_0,r_1,\ldots r_{q-1}} \left( \sum_{i=0}^{q-1}\left(\frac{1-p}{p}\right)^{r_{i}}\left(q-1\right)^{r_{i}}\right) \log\left(\sum_{i=0}^{q-1}\left(\frac{1-p}{p}\right)^{r_{i}}\left(q-1\right)^{r_{i}}\right)\\
    &= \frac{p^K}{q(q-1)^K}\left[q\left(\frac{q-1}{p}\right)^K\left( K\log\left(\frac{q-1}{p}\right)+\log(q)\right) - \Phi\right]\;,
\end{align*}
where $(1)$ follows from multinomial theorem, and 
\begin{align*}
    \Phi = \sum_{\substack{\sum_{i=0}^{q-1}r_i=K \\ r_0,r_1,\ldots,r_{q-1} \geq 0}} \binom{K}{r_0,r_1,\ldots r_{q-1}} \left(\sum_{i=0}^{q-1}\left(\frac{1-p}{p}\right)^{r_i}(q-1)^{r_i}\right)\log\left(\sum_{i=0}^{q-1}\left(\frac{1-p}{p}\right)^{r_i}(q-1)^{r_i}\right)\;.
\end{align*}

Using the expressions for $H(\bfY)$ and $H(\bfY|\bfX)$, we obtain:
\begin{align*}
     C_{\mathsf{DMS}_{q,K}(p)} &= \frac{H(\bfY) - H(\bfY|\bfX)}{\log(q)} \\
     &= \frac{p^K}{q\log(q)(q-1)^K} \left[q\left(\frac{q-1}{p}\right)^K \left( K\log\left(\frac{q-1}{p}\right) + \log(q) \right) - \Phi \right] - K H_q(p)\;.
\end{align*}
This completes the proof.
\end{proof}

\RVProbBounds*
\begin{proof}
We analyze the inner summation term in \( P(\Suc = 1) \):
\begin{align*}
    \sum_{\substack{\sum_{i=0}^{q-2}r_i=K-j \\ j > r_0,r_1,\ldots, r_{q-2} \geq 0}} 
    \binom{K-j}{r_0,r_1,\ldots, r_{q-2}}.
\end{align*}
By adjusting the range of summation, we derive the following upper and lower bounds:
\begin{align*}
    \sum_{\substack{\sum_{i=0}^{q-2}r_i=K-j \\ j > r_0,r_1,\ldots, r_{q-2} \geq 0}} 
    \binom{K-j}{r_0,r_1,\ldots, r_{q-2}}
    &\leq \sum_{\substack{\sum_{i=0}^{q-2}r_i=K-j \\  K \geq r_0,r_1,\ldots, r_{q-2} \geq 0}} 
    \binom{K-j}{r_0,r_1,\ldots, r_{q-2}} \overset{(1)}{=} (q-1)^{K-j}\;,\\
    \sum_{\substack{\sum_{i=0}^{q-2}r_i=K-j \\ j > r_0,r_1,\ldots, r_{q-2} \geq 0}} 
    \binom{K-j}{r_0,r_1,\ldots, r_{q-2}}
    &\geq \sum_{\substack{\sum_{i=0}^{q-2}r_i=K-j \\  2 > r_0,r_1,\ldots, r_{q-2} \geq 0}} 
    \binom{K-j}{r_0,r_1,\ldots, r_{q-2}} \overset{(2)}{=} \binom{q-1}{K-j}(K-j)!\;,
\end{align*}
where \( (1) \) follows from the multinomial theorem and \( (2) \) follows from a simple combinatorial argument. Using these bounds, we establish the following upper and lower bounds on \( P(\Suc = 1) \). We first derive the upper bound: 
\begin{align*}
    P(\Suc = 1) &=\sum_{j = 2}^{K} \binom{K}{j} (1-p)^j \left(\frac{p}{q-1}\right)^{K-j}
    \left(\sum_{\substack{\sum_{i=0}^{q-2}r_i=K-j \\ j > r_0,r_1,\ldots, r_{q-2} \geq 0}} 
    \binom{K-j}{r_0,r_1,\ldots, r_{q-2}}\right)\\
    &< \sum_{j = 2}^{K} \binom{K}{j}(1-p)^j p^{K-j} \\
    &=  1 - p^K - K(1-p)p^{K-1}\;.
\end{align*}
We now derive the lower bound:
\begin{align*}
P(\Suc = 1) &=\sum_{j = 2}^{K} \binom{K}{j} (1-p)^j \left(\frac{p}{q-1}\right)^{K-j}
    \left(\sum_{\substack{\sum_{i=0}^{q-2}r_i=K-j \\ j > r_0,r_1,\ldots, r_{q-2} \geq 0}} 
    \binom{K-j}{r_0,r_1,\ldots, r_{q-2}}\right)\\
&>\sum_{j = 2}^{K} \binom{K}{j}(1-p)^j \left(\frac{p}{q-1}\right)^{K-j}\binom{q-1}{K-j} (K-j)!\;\\
&=\sum_{j = 2}^{K} \binom{K}{j}(1-p)^j \left(\frac{p}{q-1}\right)^{K-j}\left(\prod_{t=0}^{K-j-1}(q-1-t)\right)\\
&\overset{(3)}{>}\sum_{j = 2}^{K} \binom{K}{j}(1-p)^j \left(\frac{p(q-K)}{q-1}\right)^{K-j}\\
&= \left(1-\frac{p(K-1)}{q-1}\right) - \left(\frac{p(q-K)}{q-1}\right)^{K} - K(1-p)\left(\frac{p(q-K)}{q-1}\right)^{K-1}\;, 
\end{align*}
where \( (3) \) follows by lower bounding each term of the product by \( q-K \).

Next, we analyze the probabilities of the random variables $\Eras_{A}$ and $\Eras_{B}$. 

\noindent First, we evaluate $P(\Eras_{A} = 1)$:
\begin{align*}
    P(\Eras_{A} = 1) &= K! (1-p) \binom{q-1}{K-1} \left(\frac{p}{q-1}\right)^{K-1}\;, \\
    &= K(1-p)\left(\frac{p}{q-1}\right)^{K-1}\prod_{t=0}^{K-2}(q-1-t)\;.
\end{align*}
By bounding each term in the product between $q-1$ (upper bound) and $q-K$ (lower bound), we derive:
\begin{align*}
    K(1-p)\left(\frac{p(q-K)}{q-1}\right)^{K-1} < P(\Eras_{A} = 1) 
    < K(1-p)p^{K-1}\;. 
\end{align*}
Next, we evaluate $P(\Eras_{B} = 1)$:
\begin{align*}
    P(\Eras_{B} = 1) &= K!\binom{q-1}{K}\left(\frac{p}{q-1}\right)^{K} \\
    &= \prod_{t=0}^{K-1}(q-1-t)\left(\frac{p}{q-1}\right)^{K}\;.
\end{align*}
Again, by bounding each term in the product between $q-1$ and $q-K$, we obtain:
\begin{align*}
   \left(\frac{p(q-K)}{q-1}\right)^{K} < P(\Eras_{B} = 1) < p^K\;.
\end{align*}

Since the total probability sums up to $1$, i.e.,
\[
P(\Err = 1) + P(\Eras_{A} = 1) + P(\Eras_{B} = 1) + P(\Suc = 1) = 1\;,
\]
we obtain the following bounds on $P(\Err = 1) $:
\begin{align*}
     0 < P(\Err = 1 ) < \frac{p(K-1)}{q-1}\;.
\end{align*}
\end{proof}

\RVProbBoundsEps*
\begin{proof}
From the previously derived bound in Lemma~\ref{lem:rv_prob_bounds}:
\[
P(\Err = 1) < \frac{p(K-1)}{q-1} = \frac{\epsilon_q p}{K}\;.
\]

We now establish a lower bound on $P(\Suc = 1)$. From the previously derived bound in Lemma~\ref{lem:rv_prob_bounds}:
\[
P(\Suc = 1) > \left(1 - \frac{p(K-1)}{q-1}\right) - \left(\frac{p(q-K)}{q-1}\right)^{K} - K(1-p)\left(\frac{p(q-K)}{q-1}\right)^{K-1}\;.
\]
Since $K>1$, we obtain the following lower bound:
\begin{align*}
    P(\Suc = 1) 
    &> 1 - \frac{p(K-1)}{q-1} - p^{K} - K(1-p)p^{K-1}\\ 
    &= 1 - p^{K} - K(1-p)p^{K-1} - \frac{\epsilon_q p}{K}\;.
\end{align*}

To derive lower bounds for $P(\Eras_A = 1)$ and $P(\Eras_B = 1)$, we first consider:
\[
\left(\frac{q-K}{q-1}\right)^{a} = \left(1 - \frac{K-1}{q-1}\right)^{a}\;, 
\]
where $a \leq K \in \mathbb{Z}$. Applying Bernoulli's Inequality, which states that $(1+x)^{a} > 1 + ax$ for $x\geq -1$ and $x\neq = 0$ and $a \in \mathbb{Z}_{\geq 2}$, we get:
\[
\left(1 - \frac{K-1}{q-1}\right)^a > 1 - a\frac{K-1}{q-1} \geq 1 - \epsilon_q\;.
\]
Using this bound, we obtain:
\begin{align*}
    P(\Eras_A =1) 
    &> K(1-p)\left(\frac{p(q-K)}{q-1}\right)^{K-1} \\
    &> K(1-p)p^{K-1}\left(1 - \epsilon_q\right),\\
    \intertext{and}
    P(\Eras_B =1) &> \left(\frac{p(q-K)}{q-1}\right)^{K} \\
    &> p^{K}\left(1 - \epsilon_q\right)\;.
\end{align*}
This completes the proof.
\end{proof}
\end{document}